\documentclass[12pt,leqno,letterpaper]{article}
\usepackage{mathrsfs}
\usepackage{amsmath}
\usepackage{bbm}
\usepackage{amsmath,amsthm,amssymb,amscd}
\usepackage{latexsym}
\usepackage{hyperref}
\usepackage[numbers,sort&compress]{natbib}
\usepackage{hypernat}

\usepackage{enumerate}
\usepackage[utf8]{inputenc}
\usepackage[english]{babel}
\usepackage{authblk}

\newtheorem{theorem}{Theorem}[section]
\newtheorem*{maintheorem}{Main Theorem}

\newtheorem{lemma}[theorem]{Lemma}
\newtheorem{proposition}[theorem]{Proposition}

\theoremstyle{definition} \theoremstyle{remark}
\numberwithin{equation}{section}

\newcommand{\trs}{{\rm Tr\hskip -0.31em}~}
\newcommand{\trm}{{\rm Tr\hskip -0.25em}~}
\newcommand{\trl}{{\rm Tr\hskip -0.12em}~}

\newcommand{\df}[2]{\frac{d#1}{d#2}}
 
\DeclareMathOperator{\frechetdiff}{\mathit d}
\newcommand{\fds}[1]{\frechetdiff\hskip -0.43em{#1}}
\newcommand{\fdm}[1]{\frechetdiff\hskip -0.28em{#1}}
\newcommand{\fdl}[1]{\frechetdiff\hskip -0.15em{#1}}

\begin{document}

\title{Peierls-Bogolyubov's inequality for deformed exponentials}


\author[1]{Frank Hansen}
\author[2]{Jin Liang}
\author[2]{Guanghua}

\affil[1]{\footnotesize Institute for Excellence in Higher Education, Tohoku University, Sendai, Japan}
\affil[2]{School of Mathematical Sciences, Shanghai Jiao Tong
University, Shanghai, China}
\date{April 3, 2017\\{\small Revised May 31, 2017}}

\maketitle

\begin{abstract}
We study convexity or concavity of certain trace functions for the deformed logarithmic and exponential functions, and obtain in this way new trace inequalities for deformed exponentials that may be considered as generalizations of Peierls-Bogolyubov's inequality. We use these results to improve previously known lower bounds for the Tsallis relative entropy.
\end{abstract}


\section{Introduction}
In statistical mechanics and in quantum information theory the calculation of the partition function $ \trs\exp H $ of the Hamiltonian $ H $ of a physical system is an important issue, but the computation is often difficult. However, it may be simplified by first computing a related quantity $\trs\exp A, $ where $A$ is an easier to handle component of the Hamiltonian. Usually, the Hamiltonian is written as a sum $H=A+B$ of two operators, and the Peierls-Bogolyubov inequality states that

\begin{equation}\label{PB-inequality}
\log\frac {\trs\exp(A+B)}{\trs\exp A}\geq \frac{\trm\exp(A)B}{\trs\exp A}\,,
\end{equation}
which then provides information about the difficult to calculate partition function.
We give in this paper generalizations of Peierls-Bogolyubov's inequality in terms of the so-called deformed exponential and logarithmic functions. We formulate the results for operators on a finite dimensional Hilbert space $ \mathcal H, $ but note that the results with proper modifications extend also to infinite dimensional spaces.

\begin{maintheorem}
Let $ A,B\in B(\mathcal H) $ be self-adjoint operators, and let  $ \varphi $ be a positive functional on $ B(\mathcal H). $
\begin{enumerate}

\item[(i)]  If $ -\infty <q<1 $ and $ r\ge q $ and both $ A $ and $ A+B $ are bounded from above by $ -(q-1)^{-1}, $ then
\[
\log_r\trs \exp_q(A+B)-\log_r\trm \exp_q A
\ge
\bigl(\trs \exp_q A\bigr)^
{r-2} \trm (\exp_q A)^{2-q} B.
\]

\item[(ii)] If $ -\infty< q\le 0 $ and $ r\ge q $ and both $ A $ and $ A+B $ are bounded from above by $ -(q-1)^{-1}, $ then
\[
\log_r\varphi\bigl(\exp_q(A+B)\bigr) -\log_r\varphi\bigl(\exp_q (A)\bigr)
\ge
\varphi\bigl(\exp_q (A)\bigr)^{r-2}\varphi \bigl(\fdm{}\exp_q(A)B\bigr).
\]

\item[(iii)]  If $ 1<q\le 2 $ and $ r\ge q $ and both $ A $ and $ A+B $ are bounded from below by $ -(q-1)^{-1}, $ then 
\[
\log_r\trs \exp_q(A+B)-\log_r\trs\exp_q A
\ge
(\trs\exp_q A)^{r-2} \trs(\exp_q A)^{2-q}B.
\]

\item[(iv)] If $ \frac{3}{2}\le q\le 2 $ and $ r\ge q $ and both $ A $ and $ A+B $ are bounded from below by $ -(q-1)^{-1}, $ then 
\[
\log_r\varphi\bigl( \exp_q(A+B)\bigr)-\log_r\varphi(\exp_q A)
\ge
\varphi\bigl(\exp_q A\bigr)^{r-2} \varphi\bigl(\fdm{}\exp_q(A)B\bigr) .
\]

\item[(v)]  If $ q\ge 2 $ and $ r\le q $ and both $ A $ and $ A+B $ are bounded from below by $ -(q-1)^{-1}, $ then 
\[
\log_r\varphi\bigl( \exp_q(A+B)\bigr)-\log_r\varphi(\exp_q A)
\le
\varphi\bigl(\exp_q A\bigr)^{r-2} \varphi\bigl(\fdm{}\exp_q(A)B\bigr).
\]
If in particular $ \varphi $ is the trace this inequality reduces to
\[
\log_r\trs \exp_q(A+B)-\log_r\trs\exp_q A
\le
(\trs\exp_q A)^{r-2} \trs(\exp_q A)^{2-q}B.
\]
\end{enumerate}
\end{maintheorem}

In subsection~\ref{formulae for Frechet differentials of deformed exponentials}
 we give explicit formulae for the Fréchet differential operators $ \fdm{}\exp_q(A) $ in the parameter ranges $ q\le 0 $ and $ q\ge 3/2. $   
Note that the left-hand sides in the above theorem may be written as
\[
\frac{\varphi\bigl(\exp_q(A+B)\bigr)^{r-1}-\varphi(\exp_q A)^{r-1}}{r-1}\,,
\]
where $ \varphi $ in $ (i) $ is replaced by the trace.
If we in $ (iii) $ let $ q $ tend to one, we obtain the inequality
\[
\log_r\trs \exp(A+B)-\log_r\trs\exp A\ge\frac{\trm(\exp A)B}{\bigl(\trs\exp A\bigr)^{2-r}}
\]
for $ r>1 $ and arbitrary self-adjoint operators $ A $ and $ B. $ If we furthermore let $ r $ tend to one we recover Peierls-Bogolyubov's inequality (\ref{PB-inequality}).

Furuichi \cite[Corollary 3.2]{kn:furuichi:2006} proved $ (iii) $ in the case $ r=q $ by very different methods.
It may be instructive to compare the above results with the first author's study \cite{Han15} of the deformed Golden-Thompson trace inequality.

We obtain, in Theorem~\ref{variant PB-inequality},  another variant Peierls-Bogolyubov type of inequality, and we improve,  in Theorem~\ref{lower bound of Tsallis relative entropy},
 previously known lower bounds for the Tsallis relative entropy.

 The Peierls-Bogolyubov inequality has been widely used in statistical mechanics and quantum information theory. Recently, Bikchentaev \cite{Bik11} proved that the Peierls-Bogolyubov inequality characterizes the tracial functionals among all positive functionals on a $C^*-$algebra. Moreover, Carlen and Lieb in \cite{CLie14} combined  this inequality with the Golden-Thompson inequality to discover sharp remainder terms in some quantum entropy inequalities.

\subsection{Deformed exponentials}

The deformed logarithm $\log_q$ is defined by setting
\begin{eqnarray*}
\log_q x=\left\{
\begin{array}{ll}
\displaystyle\frac{x^{q-1}-1}{q-1} \quad &q\ne 1\\[2.5ex]
\log x  &q=1
\end{array}
\right.
\end{eqnarray*}
for $x>0.$ The deformed logarithm is also denoted the $q$-logarithm. The inverse function is called the $q$-exponential. It is denoted by  $\exp_q$ and is given by the formula
\begin{eqnarray*}
\exp_q x=\left\{
\begin{array}{ll}
(x (q-1)+1)^{1/(q-1)} \quad &q\ne 1\\[1.5ex]
\exp x \quad &q=1
\end{array}
\right.
\end{eqnarray*}
for $x> -1/(q-1).$ The $ q $-logarithm is a bijection of the positive half-line onto the open interval $ (-(q-1)^{-1},\infty). $ Furthermore,
\[
\df{}{x}\log_q(x)=x^{q-2}\qquad\text{and}\qquad \df{}{x}\exp_q(x)=\exp_q(x)^{2-q}\,.
\]
Note also that
\[
\log_q x-\log_q y=\frac{x^{q-1}-y^{q-1}}{q-1}
\]
for $ x,y>0. $ If $q $ tends to one then the $q$-logarithm and the $q$-exponential functions converge, respectively, toward the logarithmic and the exponential functions.

\section{Preliminaries}

\begin{proposition}\label{proposition on homogeneity}
Let $ f $ be a real positive function defined in the cone $ B(\mathcal H)_+ $ of positive definite operators acting on a Hilbert space $ \mathcal H, $ and assume $ f $ is homogeneous of degree $ p\ne 0. $ 
\begin{enumerate}

\item[(i)] If $ f $ is convex and $ p>0, $ then $ f^{1/p} $ is convex.

\item[(ii)] If  $ f $ is convex and $ p<0, $ then $ f^{1/p} $ is concave.

\item[(iii)] If $ f $ is convex and $ p<0 $ and $ r>0, $ then $ f^r $ is convex.

\item[(iv)] If  $ f $ is concave and $ p>0, $ then $ f^{1/p} $ is concave.

\item[(v)] If $ f $ is concave and $ p<0, $  then $ f^{1/p} $ is convex.

\item[(vi)] If $ f $ is concave and $ p>0 $ and $ r<0, $ then $ f^r $ is convex.

\end{enumerate}
\end{proposition}

\begin{proof}
Assume first that $ f $ is a convex function. The level set
\[
L=\{x\in B(\mathcal H)_+\mid f(x)\le 1\}
\]
is then convex. Take $ x,y\in B(\mathcal H)_+ $ and assume $ p>0. $ Let $ c $ and $ d  $ be any choice of positive numbers such that $ f(x)^{1/p}<c $ and $ f(y)^{1/p}<d. $ We note that $ c^{-1}x, d^{-1}y\in L $ and obtain
\[
\displaystyle f(x+y)^{1/p}=(c+d)f\Bigl(\frac{c}{c+d}\cdot \frac{x}{c}+\frac{d}{c+d}\cdot\frac{y}{d}\Bigr)^{1/p}\le c+d.
\]
Therefore $ f(x+y)^{1/p}\le f(x)^{1/p}+f(y)^{1/p} $ and by homogeneity, we conclude that $ f^{1/p} $ is convex. If $ p<0 $ we choose $ c,d>0 $ such that $ f(x)^{1/p}>c $ and $ f(y)^{1/p}>d. $ This is possible since $ f $ is assumed to be positive. Since the exponent is negative we obtain $ f(x)<c^p $ and $ f(y)<d^p, $ and therefore by homogeneity 
\[
f(c^{-1}x)=c^{-p}f(x)<1\qquad\text{and}\qquad f(d^{-1}y)=d^{-p}f(y)<1.
\]
It follows that $ c^{-1}x, d^{-1}y\in L $ and thus
\[
\displaystyle f(x+y)^{1/p}=(c+d)f\Bigl(\frac{c}{c+d}\cdot \frac{x}{c}+\frac{d}{c+d}\cdot\frac{y}{d}\Bigr)^{1/p}\ge c+d,
\]
where we again used that the exponent is negative. Therefore $ f(x+y)^{1/p}\ge f(x)^{1/p}+f(y)^{1/p} $ and by homogeneity we conclude that $ f^{1/p} $ is concave. This proves $ (i) $ and $ (ii). $ Under the assumptions in $ (iii) $ we  proceed as under $ (ii) $ to obtain
\[
f(x+y)^{1/p}\ge c+d.
\]
By homogeneity and since the exponent $ rp $ is negative,  we obtain the inequality
\[
f\Bigl(\frac{x+y}{2}\Bigr)^r\le \Bigl(\frac{c+d}{2}\Bigr)^{rp}\le\frac{c^{rp}+d^{rp}}{2}
\]
implying convexity of $ f^r. $ We obtain $ (iv), (v) $ and $ (vi) $ by a variation of the reasoning used to obtain $ (i), (ii) $ and $ (iii). $
\end{proof}

\begin{proposition}\label{Main proposition 1}
Consider the function
\[
G(A)=\bigl(\trm A^{p}\bigr)^{1/r}
\]
defined in positive definite operators. Then

\begin{enumerate}
\item[(i)]   $ G $ is concave for $ r\le p< 0, $ 
\item[(ii)]  $ G $ is convex for $ p<0 $ and $ r>0, $
\item[(iii)]  $ G $ is concave for $ 0<p\le 1 $ and $ r\ge p,$
\item[(iv)]   $ G $ is convex for $ p\ge 1 $ and $ 0<r\le p. $
\item[(v)]  $ G $ is convex for $  0<p\le 1 $ and $ r<0. $

\end{enumerate}
\end{proposition}

\begin{proof}
Since the real function $ t\to t^{p} $ is convex in positive numbers for $ p\le 0 $ and $ p\ge 2 $ and concave for $ 0\le p\le 1, $ it is well known that the trace function $ A\to \trm A^p $  retains the same properties. A historic account of this result may be found in
\cite[Introduction]{kn:lieb:2002}.
By $ (ii) $ and $(i)$ in Proposition~\ref{proposition on homogeneity} we thus obtain that the function
\[
A\to\bigl(\trm A^{p}\bigr)^{1/p}
\]
is concave for $ p<0 $ and convex for $ p\ge 2. $  Furthermore, since the real function $ t\to t^{p/r} $ is concave and increasing for $ r\le p<0, $ we derive $ (i) $ in the assertion. Part $ (ii) $ then follows by Proposition~\ref{proposition on homogeneity}{$(iii),$} and Part $(iii)$ follows from Proposition~\ref{proposition on homogeneity}{$(iv)$} by noting that $ 0<p/r\le 1. $ Part $ (iv) $ follows from Proposition~\ref{proposition on homogeneity}{$(i)$} by noting that $ p/r\ge 1, $ and part $ (v)$ finally follows from Proposition~\ref{proposition on homogeneity}{$(vi)$}.
\end{proof}

Note that  $ (\trm A^p)^{1/p} $ for  $ p\ge 1 $ is the Schatten $p$-norm of the positive definite matrix $ A. $ The convexity in this case may also be derived by noting that a norm satisfies the triangle inequality and is positively homogeneous.

\begin{proposition}\label{Main proposition 2}
Let $ B\in B(\mathcal H) $ be an arbitrary operator and consider the function
\[
F(A)=\bigl(\trl B^* A^{p} B\bigr)^{1/r}
\]
defined in positive definite operators. Then

\begin{enumerate}
\item[(i)]  $ F $ is concave for $ -1\le p<0 $ and $ r\le p, $
\item[(ii)] $ F $ is convex for $ -1\le p<0 $ and $ r>0, $
\item[(iii)]  $ F $ is concave for $ 0<p\le 1 $ and $ r\ge p,$
\item[(iv)]   $ F $ is convex for $ 1\le p\le 2 $ and $ 0<r\le p, $
\item[(v)]  $ F $ is convex for $ 0<p\le 1 $ and $ r<0. $
\end{enumerate}
\end{proposition}

\begin{proof} By continuity we may assume $ BB^* $ invertible. 
Since the function $ t\to t^p $ is operator convex for $ -1\le p\le 0 $ or $ 1\le p\le 2, $ it follows that the trace function $ A\to\trl B^* A^p B $ is convex for these parameter values. It then follows by $ (ii) $ and $ (i) $ in Proposition \ref{proposition on homogeneity} that the function
\[
A\to (\trl B^* A^p B)^{1/p}
\]
is concave for $ -1\le p<0 $ and convex for $ 1\le p\le 2. $ Furthermore, since the real function $ t\to t^{p/r} $ is concave and increasing for $ r\le p<0 $ we derive part $ (i) $ of the assertion. Part $ (ii) $ then follows by  Proposition~\ref{proposition on homogeneity}{$(iii)$}. Parts $(iii)$ to $ (vi) $ now follow by minor variations of the reasoning in the preceding proposition.
\end{proof}

\subsection{Some deformed trace functions}

\begin{theorem}\label{convexity of G(A)}
Consider the function
\[
G(A)=\log_r\trm \exp_q(A)
\]
defined in self-adjoint $ A>-(q-1)^{-1} $ for $ q>1, $ and in self-adjoint $ A<-(q-1)^{-1} $ for $ q<1. $ Then

\begin{enumerate}

\item[(i)] If $ -\infty< q<1 $ and $ r\ge q, $  then $ G $ is convex,

\item[(ii)] If $ 1<q\le 2 $ and $ r\ge q, $ then $ G $ is convex,

\item[(iii)] If $ q\ge 2 $ and $ r\le q, $ then $ G $ is concave.
\end{enumerate}
\end{theorem}

\begin{proof} Note that the conditions on $ A $ ensure that $ A(q-1)+1> 0 $ for both $ q<1 $ and $ q>1. $
By calculation we obtain
\[
\begin{array}{rl}
G(A)&=\log_r\trm\exp_q(A)\\[2ex]
&=\displaystyle\frac{1}{r-1}\left(\Bigl(\trm (A(q-1)+1)^{1/(q-1)}\Bigr)^{r-1} -1\right).
\end{array}
\]
Under the assumptions in $ (i) $ we obtain
\[
\frac{1}{r-1}\le\frac{1}{q-1}<0
\]
for $ q\le r<1. $ By Proposition~\ref{Main proposition 1}{$(i)$} and since the factor $ (r-1)^{-1} $ is negative, it follows that $ G $ is convex. If $ r>1 $ then $ (1-r)^{-1}>0 $ and the convexity of $ G $  follows by Proposition~\ref{Main proposition 1}{$(ii)$}. This proves the first statement.
Under the assumptions in $ (ii) $ we obtain
\[
\frac{1}{q-1}\ge 1\qquad\text{and}\qquad 0<\frac{1}{r-1}\le \frac{1}{q-1}\,,
\]
thus $ G $ is convex by Proposition~\ref{Main proposition 1}{$(iv)$}. Under the assumptions in $ (iii) $ we first consider the case $ r>1 $ and obtain
\[
0<\frac{1}{q-1}\le 1\qquad\text{and}\qquad  \frac{1}{r-1}\ge \frac{1}{q-1}\,,
\]
thus $ G $ is concave by Proposition~\ref{Main proposition 1}{$(iii)$}. If $ r<1 $ then we use Proposition~\ref{Main proposition 1}{$(v)$} to obtain that $ (r-1)G $ is convex. Since $ r-1<0 $ we conclude that $ G $ is concave also in this case.
\end{proof}

\begin{theorem}\label{convexity (concavity) of F(A)}

Let $ B $ be arbitrary and consider the function
\[
F(A)=\log_r\trl B^*\exp_q(A)B
\]
defined in self-adjoint $ A>-(q-1)^{-1} $ for $ q>1, $ and in self-adjoint $ A\le(1-q)^{-1} $ for $ q<1. $

\begin{enumerate}

\item[(i)] If $ -\infty<q\le 0 $ and $ r\ge q, $ then $ F $ is convex,

\item[(ii)] If $ \frac{3}{2}\le q\le 2 $ and $ r\ge q, $ then $ F $ is convex,

\item[(iii)] If $ q\ge 2 $ and $ r\le q, $ then $ F $ is concave.
\end{enumerate}
\end{theorem}

\begin{proof}
By calculation we obtain
\[
\begin{array}{rl}
F(A)&=\log_r\trl B^*\exp_q(A) B\\[2ex]
&=\displaystyle\frac{1}{r-1}\left(\Bigl(\trl B^* (A(q-1)+1)^{1/(q-1)}B\Bigr)^{r-1} -1\right).
\end{array}
\]
Under the assumptions in $ (i) $ we obtain
\[
-1\le\frac{1}{q-1}<0\qquad\text{and}\qquad \frac{1}{r-1}\le\frac{1}{q-1}
\]
for $ q\le r<1. $ By Proposition~\ref{Main proposition 2}{$(i)$} and since the factor $ (r-1)^{-1} $ is negative, it follows that $ F $ is convex.  If $ r>1 $ then $ (1-r)^{-1}>0 $ and the convexity of $ F $ follows by Proposition~\ref{Main proposition 2}{$(ii)$}. This proves the first statement.
Under the assumptions in $ (ii) $ we obtain
\[
1\le\frac{1}{q-1}\le 2\qquad\text{and}\qquad\frac{1}{r-1}\le\frac{1}{q-1}\,
\]
thus $ F $ is convex by Proposition~\ref{Main proposition 2}{$(iv)$}. The last case is argued as in the preceding theorem by considering the cases $ r>1 $ and $ r<1 $ separately.
\end{proof}

Note in the above theorem there is a gap between $ 0 $ and $ 3/2 $ for the possible values of $ q. $

\section{Peierls-Bogolyubov type inequalities}

We first obtain a variant Peierls-Bogolyubov type inequality as a consequence of Proposition~\ref{Main proposition 1}.
Take positive definite operators $ A, B\in B(\mathcal H) $ and define the function
\[
g(t)= G(A+tB)=\bigl(\trl (A+tB)^{p}\bigr)^{1/r}\qquad t\in[0,1].
\]
Since $g(t)$ is convex for $ p\ge 1 $ and $ 0<r\le p$ we obtain the inequality,
\begin{equation}\label{convexity inequality 1}
g(1)-g(0)\ge\frac{g(t)-g(0)}{t} \qquad 0<t\le 1
\end{equation}
for these parameter values. By concavity we obtain the opposite inequality for the parameter values $ 0<p\le 1 $ and $ r\ge p, $ and for the parameter values for $ p<0 $ and $ r\le p<0. $

\begin{theorem}\label{variant PB-inequality} For positive definite operators $ A, B\in B(\mathcal H) $ we have
\begin{enumerate}

\item[(i)]  If $ p\ge 1 $ and $ 0<r\le p $ then
\[
\bigl(\trl (A+B)^{p}\bigr)^{1/r}-\bigl(\trm A^{p}\bigr)^{1/r}
\ge
\frac{p}{r}\bigl(\trl A^p\bigr)^{(1-r)/r}\trm A^{p-1}B.
\]

\item[(ii)]  If $ 0<p\le 1 $ and $ r\ge p $ or if  $ p<0 $ and $ r\le p<0 $ then
\[
\bigl(\trl (A+B)^{p}\bigr)^{1/r}-\bigl(\trl A^p\bigr)^{1/r}
\le
\frac{p}{r}\bigl(\trl A^p\bigr)^{(1-r)/r}\trl A^{p-1} B.
\]

\end{enumerate}
\end{theorem}

\begin{proof}
With the parameter values in $ (i) $ we may let $ t $ tend to zero in (\ref{convexity inequality 1}) and obtain the inequality $ g(1)-g(0)\ge g'(0). $ We note that
$ g(1)-g(0) $ is the left hand side in the desired inequality. Furthermore,
\[
\begin{array}{rl}
g'(0)&=\displaystyle\fdm{}\bigl(\trl A^p\bigr)^{1/r}B=\frac{1}{r}(\trl A^p)^{(1-r)/r}\fdm{}\bigl(\trl A^p\bigr)B\\[2ex]
&=\displaystyle\frac{1}{r}(\trl A^p)^{(1-r)/r}\trs \fdl{}A^p B=\frac{p}{r}(\trl A^p)^{(1-r)/r}\trl A^{p-1} B,
\end{array}
\]
where we used the chain rule for Fréchet differentiation, the linearity of the trace, and the formula in \cite[Theorem 2.2]{kn:hansen:1995}. This proves case $ (i). $ Case $ (ii) $ follows by virtually the same argument using the opposite inequality in (\ref{convexity inequality 1}).
\end{proof}

We then explore consequences of Theorem~\ref{convexity of G(A)}. 
If $ -\infty<q<1 $ we take self-adjoint operators $ A, B\in B(\mathcal{H}) $ such that both $ A $ and $ A+B $ are bounded from above by $ -(q-1)^{-1}. $  For $ t\in [0,1] $ we note that $ A+tB=(1-t)A+t(A+B)< -(q-1)^{-1} $  such that  $ (q-1)(A+tB)+1>0. $ The function
\begin{equation}\label{The function h(t)}
h(t)= \log_r\trs\exp_q(A+tB)\qquad t\in[0,1]
\end{equation}
is thus well-defined and convex for $ -\infty< q< 1 $ and $ r\ge q. $ Therefore,
\begin{equation}\label{convexity inequality 2}
h(1)-h(0)\ge\frac{h(t)-h(0)}{t} \qquad 0<t\le 1
\end{equation}
for these parameter values.

For $ q>1 $ we take self-adjoint operators $ A, B\in B(\mathcal{H}) $ such that both
$ A $ and $ A+B $ are bounded from below by $ -(q-1)^{-1}. $ For $ t\in [0,1] $ we note that $ A+tB=(1-t)A+t(A+B)> -(q-1)^{-1} $ such that $ (q-1)(A+tB)+1>0. $ The function defined in (\ref{The function h(t)}) is thus well-defined. It is convex for $ 1<q\le 2 $ and $ r\ge q, $ and it is concave for $ q\ge 2 $ and $ r\le q. $   
In the first case we thus retain the inequality in (\ref{convexity inequality 2}), while the inequality is reversed in the latter case.

\begin{theorem}\label{deformed PB-inequality 1} Let $ A, B\in B(\mathcal H) $ be self-adjoint operators. 

\begin{enumerate}

\item[(i)] If $ -\infty <q<1 $ and $ r\ge q $ and both $ A $ and $ A+B $ are bounded from above by $ -(q-1)^{-1}, $ then
\[
\log_r\trs \exp_q(A+B)-\log_r\trm \exp_q A
\ge
\bigl(\trs \exp_q A\bigr)^
{r-2} \trm (\exp_q A)^{2-q} B.
\]

\item[(ii)]  If $ 1<q\le 2 $ and $ r\ge q $ and both $ A $ and $ A+B $ are bounded from below by $ -(q-1)^{-1}, $ then
\[
\log_r\trs\exp_q(A+B)-\log_r\trs\exp_q A
\ge
\bigl(\trs \exp_q A\bigr)^{r-2}\trm (\exp_q A)^{2-q}B.
\]

\item[(iii)]  If $ q\ge 2 $ and $ r\le q $ and both $ A $ and $ A+B $ are bounded from below by $ -(q-1)^{-1}, $ then
\[
\log_r\trs \exp_q(A+B)-\log_r\trm \exp_q A
\le
\bigl(\trs \exp_q A\bigr)^
{r-2} \trm (\exp_q A)^{2-q} B.
\]
\end{enumerate}
\end{theorem}

\begin{proof}
With the parameter values in $ (i) $ we may let $ t $ tend to zero in (\ref{convexity inequality 2}) and obtain the inequality $ h(1)-h(0)\ge h'(0). $ We note that
$ h(1)-h(0) $ is the left hand side in the desired inequality. Furthermore,
\[
\begin{array}{rl}
h'(0)&=\fdm{}\bigl(\log_r\trs\exp_q A\bigr)B=(\trs\exp_qA)^{r-2}\fdm{}\bigl(\trs\exp_qA\bigr)B\\[1ex]
&=(\trs\exp_qA)^{r-2}\trm\fdm{}\exp_q(A)B=(\trs\exp_qA)^{r-2}\trs\exp_q'(A)B\\[1ex]
&=(\trs\exp_qA)^{r-2}\trm(\exp_qA)^{2-q}B,
\end{array}
\]
where we used the chain rule for Fréchet differentiation, the derivatives of the deformed logarithmic and exponential functions, the linearity of the trace, and the formula in \cite[Theorem 2.2]{kn:hansen:1995}. This proves case $ (i). $ The other cases follow by a variation of this reasoning.
\end{proof}

By a similar line of arguments as in the two previous theorems we finally obtain the following consequences of Theorem~\ref{convexity (concavity) of F(A)}.

\begin{theorem}\label{deformed PB-inequality 2}  Let $ C\in B(\mathcal H) $ be arbitrary and $ A, B\in B(\mathcal H) $ be self-adjoint. 
\begin{enumerate}

\item[(i)] If $ -\infty< q\le 0 $ and $ r\ge q $ and both $ A $ and $ A+B $ are bounded from above by $ -(q-1)^{-1}, $ then
\[
\begin{array}{l}
\log_r\trl C^*\exp_q(A+B)C -\log_r\trl C^*\exp_q (A)C\\[2ex]
\ge
\bigl(\trl C^*\exp_q (A)C \bigr)^{r-2}\trl C^*\bigl(\fdm{}\exp_q(A)B\bigr)C.
\end{array}
\]

\item[(ii)]  If $ \frac{3}{2}\le q\le 2 $ and $ r\ge q $ and both $ A $ and $ A+B $ are bounded from below by $ -(q-1)^{-1} $ then
\[
\begin{array}{l}
\log_r\trl C^*\exp_q(A+B)C - \log_r\trl C^*\exp_q (A)C\\[2ex]
\ge
\bigl(\trl C^*\exp_q (A)C \bigr)^{r-2}\trl C^*\bigl(\fdm{}\exp_q(A)B\bigr)C.
\end{array}
\]

\item[(iii)]  If $ q\ge 2 $ and $ r\le q $ and both $ A $ and $ A+B $ are bounded from below by $ -(q-1)^{-1} $ then
\[
\begin{array}{l}
\log_r\trl C^*\exp_q(A+B)C -\log_r\trl C^*\exp_q (A)C\\[2ex]
\le
\bigl(\trl C^*\exp_q (A)C \bigr)^{r-2}\trl C^*\bigl(\fdm{}\exp_q(A)B\bigr)C.
\end{array}
\]
\end{enumerate}
\end{theorem}

\begin{proof}
We follow a similar path as in the proof of Theorem~\ref{deformed PB-inequality 1} and consider the function
\begin{equation}\label{convexity inequality 3}
h(t)= \log_r\trl C^*\exp_q(A+tB)C\qquad t\in[0,1]
\end{equation}
which by Theorem~\ref{convexity (concavity) of F(A)} is convex for the parameter values in $ (i). $
We obtain by an argument similar to the one given in the proof of Theorem~\ref{deformed PB-inequality 1} that $ h(1)-h(0)\ge h'(0), $ and we note that
$ h(1)-h(0) $ is the left hand side in the desired inequality. Furthermore,
\[
\begin{array}{rl}
h'(0)&=\fdm{}\bigl(\log_r\trl C^*\exp_q (A)C\bigr)B\\[1ex]
&=(\trl C^*\exp_q (A) C)^{r-2}\fdm{}\bigl(\trl C^*\exp_q(A)C\bigr)B\\[1ex]
&=(\trl C^*\exp_q(A)C)^{r-2}\trl C^*\bigl(\fdm{}\exp_q(A)B\bigr)C
\end{array}
\]
where we used the chain rule for Fréchet differentiation, the derivative of the deformed logarithmic function, and the linearity of the trace. This proves case $ (i). $ Since the function $ h $ in (\ref{convexity inequality 3}) is convex for the parameter values in $ (ii) $ and concave for the parameter values in $ (iii) $ these cases follow by virtually the same line of arguments as in $ (i). $
\end{proof}

Note that $ (iii) $ in Theorem~\ref{deformed PB-inequality 2} is a generalization of $ (iii) $ in Theorem~\ref{deformed PB-inequality 1}. Since $ C $ is arbitrary in the above theorem we may replace the trace by any other positive functional on $ B(\mathcal H). $ The main theorem now follows from Theorem~\ref{deformed PB-inequality 1} and Theorem~\ref{deformed PB-inequality 2}.

\section{The Tsallis relative entropy}

In this section we study lower bounds for  the (generalized) Tsallis relative entropy. For basic information about the Tsallis entropy and the Tsallis relative entropy we refer the reader to references \cite{Tsa88,Tsa09}.

The Tsallis relative entropy $ D_{p}(X\mid Y) $ is for positive definite operators $ X,Y\in B(\mathcal H)  $ and $ p\in[0,1) $ defined by setting
\[
D_{p}(X\mid Y)= \frac{\trm (X-X^{p}Y^{1-p})}{1-p}=\trl X^{p}(\log_{2-p}X-\log_{2-p}Y).
\]
By letting $ p $ tend to one this expression converges to the relative quantum entropy
\[
U(X\mid Y)=\trl X(\log X-\log Y)
\]
introduced by Umegaki \cite{Ume62}. It is known \cite[Proposition 2.4]{FYK04} that the Tsallis relative entropy is non-negative for states. This also follows directly from the following:

\begin{lemma}\label{trace inequality for states}
Let $ \rho $ and $ \sigma $ be states. Then
\[
\trl \rho^{1-p}\sigma^p\le 1
\]
for $ 0\le p\le 1. $
\end{lemma}

\begin{proof}
Consider states $ \rho $ and $ \sigma, $ and let
$ E\subseteq [0,1] $ be the set of exponents $ p $ such that $ \trs \rho^{1-p}\sigma^p\le 1. $ We take $ p,q\in E  $ and obtain
\[
\begin{array}{l}
\trl \rho^{1-(p+q)/2}\sigma^{(p+q)/2}=\trl\rho^{(1-p)/2}\rho^{(1-q)/2}\sigma^{p/2}\sigma^{q/2}\\[1.5ex]
=\trl\sigma^{p/2}\rho^{(1-p)/2}\rho^{(1-q)/2}\sigma^{q/2}=\trm(\rho^{(1-p)/2}\sigma^{p/2})^* \rho^{(1-q)/2}\sigma^{q/2}\\[1.5ex]
\le\bigl(\trm (\rho^{(1-p)/2}\sigma^{p/2})^* \rho^{(1-p)/2}\sigma^{p/2}\bigr)^{1/2}
\bigl(\trm (\rho^{(1-q)/2}\sigma^{q/2})^* \rho^{(1-q)/2}\sigma^{q/2}\bigr)^{1/2}\\[1.5ex]
=\bigl(\trl \rho^{1-p}\sigma^p\bigr)^{1/2} \bigl(\trl \rho^{1-q}\sigma^q\bigr)^{1/2}\le 1,
\end{array}
\]
where we used Cauchy-Schwarz' inequality. This
shows that $ E $ is midpoint-convex. Since $ E $ also is closed and $ 0,1\in E, $ we conclude that $ E=[0,1]. $
\end{proof}

\begin{theorem}\label{lower bound of Tsallis relative entropy}
 Let $ q\in (0,1]$ and take $p\leq q.$  Then, for positive definite operators $X, Y\in B(\mathcal H),$  the inequality
\begin{eqnarray*}
 \frac{\trl  X-(\trl X)^{p}(\trl Y)^{1-p}}{1-p}
\leq
D_q(X\mid Y)
\end{eqnarray*}
is valid, 
where by convention $ D_1(X\mid Y)=U(X\mid Y). $
\end{theorem}

\begin{proof}
Let $ X,Y\in B(\mathcal H) $ be positive definite operators and take $ 1<q\le 2 $ and $ r\ge q. $ By setting
\[
A=\log_q X\qquad\text{and}\qquad B=\log_q Y-\log_qX
\]
we obtain self-adjoint $ A,B $ such that both $ A $ and $ A+B $ are bounded from below by $ -(q-1)^{-1}. $ We may thus apply $ (i) $ of  Theorem~\ref{deformed PB-inequality 1} and obtain after a little calculation the inequality
\[
\frac{\trl X-(\trl X)^{2-r}(\trl Y)^{r-1}}{r-1}
\le
\trl X^{2-q}(\log_{q}X-\log_{q}Y).
\]
By setting $ p=2-r $ and renaming $ q $ by $ 2-q $ we obtain the stated inequality for $ q\in (0,1] $ and $ p\le q. $
\end{proof}

 The lower bound of the Tsallis relative entropy $ D_q(X\mid Y)$  in Theorem~\ref{lower bound of Tsallis relative entropy}  was obtained in \cite[Theorem 3.3]{FYK04} in the special case  $p=q.$ The family of lower bounds given above is in general not an increasing function in the parameter $ p $ and may therefore, depending on $ \trl X $ and $ \trl Y, $ provide better lower bounds.

\section{Various Fréchet differentials}

In order to obtain a more detailed understanding of the bounds obtained in the Main Theorem we need to provide explicit formulae for the Fréchet differential operator $ \fdm{}\exp_q $ in the parameter range
$ q\ge 3/2. $ 
The integral representation
\begin{equation}\label{integral representation of t^p for 0<p<1}
t^p=\frac{\sin p\pi}{\pi}\int_0^\infty \frac{t}{t+\lambda}\lambda^{p-1}\,d\lambda\qquad t>0
\end{equation}
valid for $ 0<p<1 $ is well-known. Since $ t\to t^p $ is operator monotone the representation may be quite easily derived by calculating the representing measure, see for example \cite[Theorem 5.5]{kn:hansen:2013:1}. Furthermore,
since by an elementary calculation
\[
\fds{}\left(\frac{x}{x+\lambda}\right)h=\lambda(x+\lambda)^{-1}h(x+\lambda)^{-1},
\]
we obtain the integral representation
\begin{equation}\label{diff of x^p for 0<p<1}
\fdm{}(x^p) h
=\frac{\sin p\pi}{\pi}\int_0^\infty (x+\lambda)^{-1}h(x+\lambda)^{-1}\lambda^p\,d\lambda,\qquad 0<p<1
\end{equation}
valid for positive definite $ x. $ Since by (\ref{integral representation of t^p for 0<p<1}) we have
\begin{equation}\label{integral representation of t^p for -1<p<0}
t^{p-1}=\frac{\sin p\pi}{\pi}\int_0^\infty \frac{1}{t+\lambda}\lambda^{p-1}\,d\lambda\qquad t>0
\end{equation}
for $ 0<p<1 $ and
\[
\fds{}\left(\frac{1}{x+\lambda}\right)h=(x+\lambda)^{-1}h(x+\lambda)^{-1},
\]
we obtain the integral representation
\begin{equation}\label{diff of x^p for -1<p<0}
\fdm{}(x^p) h
=\frac{\sin (p+1)\pi}{\pi}\int_0^\infty (x+\lambda)^{-1}h(x+\lambda)^{-1}\lambda^p\,d\lambda,\qquad -1<p<0
\end{equation}
valid for positive definite $ x. $
By using the rule for the Fréchet differential of a product, or by an elementary direct calculation, we obtain the general identity
\begin{equation}\label{diff of x^q for 1<q<2}
\fdm{}(x^{p+1})h= hx^p+x\fdm{}(x^p)h,
\end{equation}
which combined with (\ref{diff of x^p for 0<p<1}) provides a formula for the Fréchet differential of $ x^p $ for $ 1<p<2. $
If $ h $ is self-adjoint the formula in (\ref{diff of x^q for 1<q<2}) may be written on the form
\[
\fdm{}(x^{p+1})h=\frac{hx^p+x^ph}{2}+\fdm{}(x^p)\frac{xh+hx}{2}
\]
which is then manifestly self-adjoint.

\subsection{The deformed logarithm}

By setting $ t=1 $ in (\ref{integral representation of t^p for 0<p<1}) we obtain
\[
1=\frac{\sin p\pi}{\pi}\int_0^\infty \frac{\lambda^{p-1}}{1+\lambda}\,d\lambda
\]
and thus
\[
t^p-1=\frac{\sin p\pi}{\pi}\int_0^\infty \left(\frac{t(1+\lambda)}{t+\lambda}-1\right)\frac{\lambda^{p-1}}{1+\lambda}\,d\lambda
=\frac{\sin p\pi}{\pi}\int_0^\infty  \frac{t-1}{t+\lambda}\,\frac{\lambda^p}{1+\lambda}\,d\lambda
\]
for $ 0<p<1 $ and $ t>0. $
We therefore obtain the following integral representation of the deformed logarithm
\begin{equation}\label{integral representation of the deformed logarithm}
\log_q t=\frac{t^{q-1}-1}{q-1}=\frac{\sin(q-1)\pi}{(q-1)\pi}\int_0^\infty  \frac{t-1}{t+\lambda}\,\frac{\lambda^{q-1}}{1+\lambda}\,d\lambda\qquad t>0
\end{equation}
valid for $ 1<q<2. $ Since by an elementary calculation
\[
\fds{}\left(\frac{x-1}{x+\lambda}\right)h=(1+\lambda)(x+\lambda)^{-1}h(x+\lambda)^{-1}
\]
we derive the formula
\begin{equation}\label{Frechet derivative of deformed logarithm}
\fdl\log_q(x)h=\frac{\sin (q-1)\pi}{(q-1)\pi}\int_0^\infty (x+\lambda)^{-1}h(x+\lambda)^{-1}\lambda^{q-1}\,d\lambda
\end{equation}
valid for positive definite $ x $ and $ 1<q<2. $ Note that
\begin{equation}
\fdl\log_q(x)h=\frac{1}{q-1}\fdm{}(x^{q-1}) h
\end{equation}
for all $ q>1 $ by the definition of the deformed logarithm.
If we in formula (\ref{Frechet derivative of deformed logarithm}) let $ q $ tend to $ 1 $ we obtain
\[
\fdl\log(x)h=\int_0^\infty (x+\lambda)^{-1}h(x+\lambda)^{-1}\,d\lambda
\]
as expected. If we instead set $ h=1, $ we recover the classical integral
\[
t^{q-2}=\frac{\sin(q-1)\pi}{(q-1)\pi}\int_0^\infty \frac{\lambda^{q-1}}{(t+\lambda)^2}\,d\lambda
\]
valid for $ t>0 $ and $ 1<q<2. $

\subsection{The deformed exponential}\label{formulae for Frechet differentials of deformed exponentials}

We next derive integral representations for the deformed exponential in the parameter interval $ q\ge 3/2. $ We first note that
\[
\exp_q\Bigl(t-\frac{1}{q-1}\Bigr)=\Bigl((t-\frac{1}{q-1})(q-1)+1\Bigr)^{1/(q-1)}=\bigl((q-1)t\bigr)^{1/(q-1)}
\]
for $ t>0. $ Therefore,
\begin{equation}
\fdm{}\exp_q\Bigl(x-\frac{1}{q-1}\Bigr)h=(q-1)^{1/(q-1)}\fdm{}(x^{1/(q-1)})h
\end{equation}
for positive definite $ x. $ We divide the analysis into four cases:

\begin{enumerate}

\item If $ q<0 $ then 
\[
-1<\frac{1}{q-1}<0
\]
and we may therefore calculate $ \fdm{}(x^{1/(q-1)})h $ by  the formula in (\ref{diff of x^p for -1<p<0}).

\item If $ q=\frac{3}{2} $ then $ (q-1)^{-1}=2, $ thus
\[
\exp_{3/2}(x-2)=\frac{x^2}{4}\qquad\text{and}\qquad\fdm{}\exp_{3/2}(x-2)h=\frac{xh+hx}{4}\,.
\]

\item If $ \frac{3}{2}<q<2 $ then we have
\[
\frac{1}{q-1}=p+1\qquad\text{for some}\quad p\in(0,1). 
\]
We may therefore calculate $ \fdm{}(x^{1/(q-1)})h $ by  the formulae in (\ref{diff of x^q for 1<q<2}) 
and (\ref{diff of x^p for 0<p<1}).

\item If $ q=2 $ then $ (q-1)^{-1}=1, $ thus
\[
\exp_2(x-1)=x\qquad\text{and}\qquad\fdm{}\exp_2(x-1)h=h.
\]

\item If $ q> 2 $ then
\[
0<\frac{1}{q-1}=p<1
\]
and we may calculate  $ \fdm{}(x^{1/(q-1)})h $ by the formula in (\ref{diff of x^p for 0<p<1}).

\end{enumerate}

Note that we for any $ q>1 $ and  $ x>-(q-1)^{-1} $ have the identity
\[
\trs\fdm{}\exp_q (x)h=\trs \exp_q(x)^{2-q}h.
\]
Likewise,
\[
\fdm{}\exp_q (x)h=\exp_q(x)^{2-q}h
\]
for commuting $ x $ and $ h. $

\subsection*{Acknowledgments}

The authors would like to thank the anonymous referees for helpful suggestions.
The first author acknowledges support by the Japanese Grant-in-Aid for scientific research 17K05267 and by the National Science Foundation of China 11301025. The second and third authors acknowledge support from the National Science Foundation of China 11571229.




\begin{thebibliography}{16}


\bibitem{Bik11} A. M. Bikchentaev,  The Peierls-Bogoliubov inequality in C*-algebras and characterization of tracial functionals. Lobachevskii J. Math.  32(3): 175-179, 2011.

\bibitem{CLie14} E. A. Carlen, E. H. Lieb, Remainder terms for some quantum entropy inequalities. J. Math. Phys.  55(4): 042201, 2014.





\bibitem{FYK04} S. Furuichi, K. Yanagi, K. Kuriyama, Fundamental properties of Tsallis relative entropy. J. Math. Phys. 45(12): 4868-4877, 2004.


\bibitem{kn:furuichi:2006}
S.~Furuichi.
\newblock Trace inequalities in non-extensive statistical mechanics.
\newblock {\em Linear Algebra and Its Applications} 418:821--827, 2006.


\bibitem{kn:hansen:1995}
F.~Hansen and G.K. Pedersen.
\newblock Perturbation formulas for traces on $ \uppercase{C}^* $-algebras.
\newblock {\em Publ. RIMS, Kyoto Univ.} 31:169--178, 1995.

\bibitem{kn:hansen:2013:1}
F.~Hansen.
\newblock The fast track to \uppercase{L}{\"o}wner's theorem.
\newblock {\em Linear Algebra Appl.} 438:4557--4571, 2013.


\bibitem{Han15} F. Hansen, Golden-Thompson's inequality for deformed exponentials. J. Stat. Phys. 159(5): 1300-1305, 2015.

\bibitem{kn:lieb:2002}
	 {E. Lieb and G.K.  Pedersen},
	{Convex multivariable trace functions},
	{Reviews in Mathematical Physics }{14:} {631--648}, {2002.}


\bibitem{Tsa88} C. Tsallis, Possible generalization of Bolzmann-Gibbs statistics. J. Stat. Phys. 52, 479-487, 1988.

\bibitem{Tsa09} C. Tsallis, Nonadditive entropy and nonextensive statistical mechanics - an overview after 20 years. Brazilian J Phys 39: 337-356, 2009.


\bibitem{Ume62} H. Umegaki, Conditional expectation in an operator algebra, IV (entropy and information). Kodai Math. Sem. Rep. 14: 59-85, 1962.


\end{thebibliography}
\end{document}